\documentclass[10pt]{article}
\usepackage{amsmath,amssymb}
\usepackage{amscd}
\usepackage{epsfig}
\usepackage{amsthm}
\allowdisplaybreaks
\newtheorem{thm}{Theorem}[section]

\newtheorem{defin}[thm]{Definition}
\newtheorem{prop}{Proposition}[section]
\makeindex
\begin{document}
\title{{\bf Minimal Economic Distributed Computing}}
\author{S. Youssef$^1$, J. Brunelle$^2$, J. Huth$^3$, D.C. Parkes$^4$, M. Seltzer$^4$ and J. Shank$^1$ \\
  ${}^1$Center for Computational Science, Boston University \\
  ${}^2$Faculty of Arts and Sciences IT, Harvard University \\
  ${}^3$Faculty of Arts and Sciences, Harvard University \\
  ${}^4$School of Engineering and Applied Sciences, Harvard University \\
}
\maketitle
\begin{abstract}
In an ideal distributed computing infrastructure, users would be able to use
diverse distributed computing resources in a simple coherent way, with
guaranteed security and efficient use of shared resources in accordance
with the wishes of the owners of the resources. Our strategy 
for approaching this ideal is to first find the simplest
structure within which these goals can plausibly be achieved.
This structure, we find, is given by a particular recursive distributive lattice freely constructed
from a presumed partially ordered set of all data in the infrastructure. 
Minor syntactic adjustments to the resulting algebra yields a simple language 
resembling a UNIX shell, a concept of execution and an interprocess protocol.  
Persons, organizations and servers within the system express their interests 
explicitly via a hierarchical
currency.  The currency provides a common 
framework for treating authentication, access control and
resource sharing as economic problems while also introducing a 
new dimension for improving the infrastructure over time by designing
system components which compete with each other to earn the currency.
We explain these results, discuss experience with an implementation 
called {\it egg} and point out areas where more research is needed.
\end{abstract}
%%%%%%%%%%%%%%%%%%%%%%%%%%%%%%%%%%%%%%%%%%%%%%%%%%%%%%%%%%%%%%%%%%%%%%%%

\section{Introduction}

    An ideal distributed computing infrastructure would present diverse distributed
computing resources in a conceptually coherent framework while 
guaranteeing flexible, secure, efficient shared use of the resources, 
all in accordance with the wishes of the owners of the resources.
Although these goals have an obvious attraction
and importance, the all-encompassing nature of the problem makes it 
difficult to approach as a whole.  Our strategy here is to attempt to identify 
minimal required properties of the infrastructure until a simplest
solution is essentially uniquely determined.  We then hope that 
this solution is also sufficient for distributed computing in general.
Starting with the simplest possibility, we consider an infrastructure
made from objects of a single type called ``caches.''
Taking storage of caches as a minimal requirement,
we find that caches must be a particular recursive distributive
lattice freely constructed from an assumed partially ordered
set of all data which might be stored in the system.  Minor syntactic 
re-arrangement of the resulting cache algebra yields a language 
and a single user environment which resembles a UNIX shell where ``shell commands'' 
pipe cache streams rather than byte streams.  The algebra comes with a built-in concept 
of execution and a built in universal interprocess protocol which is useful on a 
large scale where the infrastructure is inevitably made of independent communicating processes.
Within the cache-based infrastructure, authentication, resource allocation and 
access control are treated as economic transactions where an explicit hierarchical
currency is exchanged.  Although an economic solution to these problems is not 
dictated as the unique simplest solution in the same sense that caches are, we 
argue below that a substantial simplification is obtained compared with 
conventional alternatives.  At the same time, the currency opens a new dimension for
system improvement over time: caches can be designed to compete with each other to 
earn the currencies so that cache ``self interest'' results in improved system performance.

   In order to test whether caches, the language and the currency really are sufficient for 
distributed computing in general, we have implemented these structures in a system called {\it egg}\cite{Egg}.  
Experience with egg shows (section 6) that our hopes are justified in the sense that this framework 
appears to be suitable for everything that one normally wants a distributed computing infrastructure to do.  The status of 
egg as the ``simplest distributed computing infrastructure'' in the sense explained below gives some indication 
that cache-based infrastructure will be more flexible, easier to comprehend by users and easier to
make secure than alternative systems.  In the sections which follow, we define the 
infrastructure, discuss experience with egg and point out areas where more work is needed.

\section{Caches and data}

     Taking the simplest possibility first, consider an infrastructure where
all functional elements are objects of a single type and let us agree to call
these objects {\it caches}.  Since any distributed computing infrastructure must at least include
storage, storage of caches in other caches must be possible.  This means that at 
least two cache-valued binary operators on caches must exist: one to {\bf store} a cache 
in another cache and a second to {\bf retrieve} a stored cache by providing a cache as 
specification. These operations are related to each other by formulas which follow 
from what one intuitively means by ``storage.''  For example, one expects
\begin{equation}
{\bf retrieve}(A,{\bf store}(A,B))=A
\end{equation}
since, if you store cache $A$ in $B$, retrieving $A$ should again produce $A$.
Similarly, we expect that storing $A$ in $B$ followed by storing $C$ in the result should
be the same as storing $C$ in $B$ followed by storing $A$ in the result and so we expect that {\bf store}
is associative.  Continuing this way, one finds that our expectations are exactly met 
provided that caches are a distributive lattice\cite{lattice} with {\bf store} and {\bf retrieve} as join and meet respectively.
At this point, we know that 
caches must form a distributive lattice, but which one?
A clue comes from considering the partially ordered set $D$ of all data
that one might want to ever be stored in the infrastructure.  Since elements
of $D$ are to be stored, we want a distributive lattice to be constructed from $D$
in a way which ``preserves $D$ and its partial ordering.''  This informally describes what
is known as a ``free construction'' in category theory.  Limiting ourselves
to such free constructions is usefully constraining since free constructions are unique
in any category\cite{category}.  We are thus lead to search for the free distributive 
lattice constructed from a given partially ordered set $P$.  This 
is given by the set of antichain\cite{antichain} subsets of $P$ with join and 
meet defined by
\begin{equation}
A \vee B\stackrel{d}{=}(A \cup B)^{\uparrow}
\end{equation}
\begin{equation}
A \wedge B\stackrel{d}{=}(A^{\downarrow} \cap B^{\downarrow})^{\uparrow}
\end{equation}
where $A^{\uparrow}$(an antichain) is the set of elements which are maximal in $A$\cite{uparrow} and
$A^{\downarrow}$ is the set of elements in $P$ which are less than 
or equal to some element of $A$.

\begin{prop} Antichain subsets of $P$ with joint and meet as defined is a distributive lattice.
\end{prop}
\begin{proof}  First, let $A$, $B$ and $C$ be subsets of $P$ and note that $(A^{\uparrow})^{\downarrow}=A^{\downarrow}$, 
$(A^{\downarrow})^{\uparrow}=A^{\uparrow}$, $(A\cup B)^{\uparrow}=(A^{\uparrow}\cup B^{\uparrow})^{\uparrow}$,
$(A\cup B)^{\downarrow}=(A^{\downarrow}\cup B^{\downarrow})$ and $(A\cap B)^{\downarrow}=(A^{\downarrow}\cap B^{\downarrow})$.
Meet and join are obviously commutative and idempotent.  They are also associative, since, for antichains $A,B,C\subset P$,
$A\vee(B\vee C)=(A \cup B \cup C)^{\uparrow}$ and $A\wedge(B\wedge C)=(A^\downarrow \cap B^\downarrow \cap C^\downarrow)^\uparrow$.
Absorption laws follow since $A\vee(A\wedge B)=(A\cup(A^\downarrow \cap B^\downarrow)^\uparrow)^\uparrow=((A^\downarrow)^\uparrow \cup
(A^\downarrow\cap B^\downarrow)^\uparrow)^\uparrow=A^{\downarrow\uparrow}=A^\uparrow=A$ and, similarly, $A\wedge(A\vee B)=A$.
Lastly, $A\wedge(B\vee C)=(A^\downarrow\cap((B\cup C)^\uparrow)^\downarrow)^\uparrow=(A^\downarrow\cap(B\cup C)^\downarrow)^\uparrow=(A\wedge
B)\vee(A\wedge C)$ and we have a distributive lattice. \end{proof}

\noindent In addition to being a distributive lattice, the lattice of proposition 2.1 is also the 
free distributive lattice from the category of partially ordered sets and order preserving functions to the category of 
distributive lattices with $\vee$--preserving morphisms [$\phi(x\vee y)=\phi(x)\vee \phi(y)$].  
To see this, let $\mathcal{F}(P)$ be the lattice of the proposition, let $\alpha:P\rightarrow \mathcal{F}(P)$ map $p$ to $\{p\}$ and suppose that $f:P\rightarrow M$ is an order preserving map to some 
distributive lattice $M$.  Then  $F(A)\stackrel{d}{=}\vee_{a\in A}f(a)$ is the unique $\vee$-morphism such that $F\circ\alpha=f$.  As is true 
in all categories, $\mathcal{F}(P)$ is the unique distributive lattice having these properties.

The most straight-forward use of the free construction would be to let caches be antichain subsets of $D$ itself.  
Although this does produce a distributive lattice, the resulting caches cannot easily represent, for instance, a hierarchical 
file system.  The simplest plausible solution is then to define caches $C$ recursively as, roughly speaking, the free distributive 
lattice on the partially ordered set $D\times C$ with $(d,c)\leq(d',c')$ if and only if $d\leq d'$ and $c\leq c'$.
More carefully, let $C_o\stackrel{d}{=}\{\{\}\}$ and, for $n\geq 1$, let 
$C_n\stackrel{d}{=}\mathcal{F}(D\times C_{n-1})$ so that $C_n$ are caches with depth less than or equal to $n$. 
Noting that $C_o\subset C_1\subset\dots$, we can define caches in general as follows.
\begin{defin} Caches are the distributive lattice $\bigcup_{n=0}^{\infty}C_n$.\end{defin}
\noindent It is sometimes convenient to use algebraic notation for caches, writing, for example,
\begin{equation}
(a,X) \vee (b,Y) \vee (c,Z)
\end{equation}
for the antichain $\{(a,X), (b,Y), (c,Z)\}$ and writing ``$0$'' for the
minimum empty antichain. Generally speaking, the fewer the minimal elements in $D$, 
the more caches tend to collapse.  For example, if $D$
is the integers with their usual ordering (so that no integer is minimal), 
then $(1,0)$, $(2,0)$ and $(1,(2,0))\vee (2,0)$ are caches and
$(1,0)\wedge[((1,(2,0))\vee(2,0)]=[(1,0)\wedge(1,(2,0))]\vee[(1,0)\wedge(2,0)]=(1,0)\vee(1,0)=(1,0)$. 
On the other hand, if $D$ are strings with trivial ordering (so that all strings are minimal), then caches 
like $(hello,0)\vee(world,0)$ cannot be reduced.  Non-minimal elements of $D$ are 
useful for selection, for example if $D$ is the set of strings with $world\leq wor*$, then 
$[(hello,0)\vee(world,0)]\wedge (wor*,0)=(world,0)$. 
With these examples in mind, the data $D$ used for the infrastructure should have many minimal elements to provide
a wide variety of stable caches and a rich order structure to provide flexible selection
and database-like functionality.  As we shall see, it is advantageous for $D$ to allow lazy refinement and to allow
new kinds of data to be dynamically included in $D$ as they are encountered.  Of course, easy user 
specification of elements of $D$ is also an important criterion.

With simple user specification in mind, construct $D$ from a meet semilattice direct sum\cite{meetsemidirectsum}
\begin{equation}
\Delta = d_0\oplus d_1\oplus\dots\oplus d_n
\end{equation}
of meet semilattice {\it basic data types} $d_0,\dots,d_n$ with minimum elements $0_0,0_1,\dots,0_n$ 
respectively. 
A user can thus specify an element of $\Delta$ by providing zero or more type/value pairs, such as
{\tt text:hello}, {\tt size:5}, {\tt \{text:hello, size:5\}} or {\tt http:www.bu.edu},
where {\bf text}, {\bf size} and {\bf http} are meet semilattices.  Since any partial
order can be made into a meet semilattice by adding a minimum element if necessary, the basic 
data types only need to be sets with some conveniently chosen partial order.

     Before completing the construction of $D$ from $\Delta$, we need some preliminary results 
concerning ``extensions'' of meet semilattices which will provide lazy refinement of elements of $\Delta$.

\begin{defin} An extension of meet-semilattice S is a mapping $F(s)=s\wedge f(s)$
where $f:S\rightarrow S$ is an order preserving function.
\end{defin}

\noindent It is easy to see that extensions are order preserving, non-increasing and closed under composition.
Given a set $E=\{F_0,F_1,\dots,F_n\}$ of extensions of meet semilattice $\Delta$, $d\in \Delta$ is a {\it fixed point}
of $E$ if $F(d)=d$ for all $F$ in $E$.  Such fixed points are unique, for suppose that $F(d)$ and $G(d)$ are
fixed points of $d$.  Then $G(F(d))=F(d)=d\wedge f(d) = d\wedge f(d)\wedge g(d\wedge f(d))\leq d\wedge f(d) \wedge g(d) \wedge g(f(d))
\leq d\wedge g(d)$.  Therefore $F(d)\leq G(d)$.  Similarly $G(d)\leq F(d)$.

   With the above results, suppose that we have a set $E$ of $\Delta$ extensions and
a guarantee that every $d\in \Delta$ reaches a fixed point of $E$ after a 
finite number of extensions.  Since fixed points are unique, we can let elements
of $\Delta$ be equivalent if they have the same fixed point and let $D$ be the 
equivalence classes with
\begin{equation}
[a]\wedge[b]\stackrel{d}{=}[a\wedge b]
\end{equation}
so that $D$ is clearly a meet semilattice provided that $\wedge$ is a function.
This is guaranteed by proposition 2.2.

\begin{prop} $D$ is a meet semilattice.\end{prop}
\begin{proof}First note that for $a,b\in \Delta$, if $F(a \wedge b)$ is the fixed point of $a\wedge b$,
then $F(a\wedge b)=F(F(a)\wedge F(b))$.  This follows because $F(a\wedge b)\leq F(a)\wedge F(b)$ implies
$F(a\wedge b)\leq F(F(a)\wedge F(b))$ and $F(a)\wedge F(b)\leq a\wedge b$ implies $F(F(a)\wedge F(b))\leq F(a\wedge b)$.
Next, note that for any two points $a$ and $b$ in $\Delta$, there is a composition of extensions which takes both
$a$ and $b$ to their corresponding fixed point. To see this, let $F$ take $a$ to its fixed point and let $G$ take
$F(b)$ to its fixed point.  Then $G\circ F$ takes both $a$ and $b$ to their fixed points.  Finally, let $a\sim a'$,
$b\sim b'$ and let $F$ take $a$, $b$ and $a\wedge b$ to their respective fixed points.  Then 
$F(a'\wedge b')=F(F(a')\wedge F(b'))=F(F(a)\wedge F(b))=F(a\wedge b)$ and the meet defined above 
makes $D$ a meet semilattice.
\end{proof}

Finally, demote $D$ from a meet semilattice to a partially ordered set by removing any element 
containing a minimum element of the basic data types $d_0,d_1,\dots,d_n$.  The resulting partially 
ordered set has a single maximum $\{\}$ and, as desired, has many minimal elements consisting of direct sums of 
covers of $0_0,0_1,\dots,0_n$, one from each of $d_0,d_1,\dots,d_n$.

    This choice of $D$ has the features needed for the infrastructure including easy user 
specification, many minima and rich ordering structure for selection purposes.  At the system
level, an element of $D$ can be represented as the fixed point of an element of $d\in \Delta$.  
Since fixed points are preserved by the extensions, extensions can lazily be applied to $d$ whenever needed and 
new extensions can be dynamically added to the system without requiring recomputation of existing data. 

\section{Cache algebra}

Given any set $d_0,d_1,\dots,d_n$ of ``basic data types'' in
the form of meet semilattices, we have defined a concrete partially ordered set $D$ 
and a corresponding uniquely determined distributive lattice of caches given by definition 2.1. 
The pair structure of caches suggests defining two additional operators $/,//:C\times D\rightarrow C$ 
which ``select subcaches'' of a cache given $d\in D$. Let
\begin{equation}(A_0\vee A_1\vee\dots)/d\stackrel{d}{=}[(A_0/d)\vee (A_1/d)\vee\dots],\end{equation}
\begin{equation}(a,X)/d\stackrel{d}{=} X \wedge (d,1),\end{equation}
\begin{equation}(A_0\vee A_1\vee\dots)//d\stackrel{d}{=}[(A_0//d)\vee (A_1//d)\vee\dots]\end{equation} and
\begin{equation}(a,X)//d\stackrel{d}{=} [(a,X)\wedge(d,1)] \vee (X//d).\end{equation}

\noindent Given a cache $(a,X)$, for example, $(a,X)/\{\}=X$ is the ``contents'' of $(a,X)$.
If $D$ are strings with $world\leq wor*$ as before and $C=(a,(hello,0)\vee(world,0))$, then $C/wor*=C//wor*=(world,0)$ and
$C//\{\}=(a,(hello,0)\vee(world,0))\vee(hello,0)\vee(world,0)$.

   The action of the operators $\vee$, $\wedge$, $/$ and $//$ is completely prescribed by 
the definitions given here and can be implemented once for all caches.  
The {\it cache algebra} is completed by adding one more binary operator $<:C\times C\rightarrow C$ 
(called ``put,'' not to be confused with ``less than'') which, on the other hand, has more flexibility.  
Put is required to distribute over $\vee$ from the right as in $(A \vee B \vee C)<Y=(A<Y)\vee (B<Y)\vee(C<Y)$ so 
that its action is determined by its action on singleton caches which is defined by
\begin{equation}(d,X)<Y\stackrel{d}{=}(d,{\bf put}(d,X,Y))\end{equation}
where {\bf put} is only required to be a cache valued function of its arguments.  It is convenient to organize
the action of {\bf put} by introducing a data type {\bf type} which specifies a ``singleton cache subtype'' in the
sense that singleton caches with different {\bf type} values have different actions under put.  
A storage cache, for instance, might have {\bf put}$({\rm type:storage},X,Y)=X\vee Y$ while a counting
cache might have {\bf put}$({\rm type:counter},X,Y)=({\rm count}:\#Y,0)$ counting the number of 
singleton caches in $Y$.    The egg infrastructure has hundreds of such singleton cache subtypes performing
many different tasks.  Since put is not generally associative, 
we take $a<b<c\stackrel{d}{=}(a < (b < c))$ grouping from the right in the absence of parenthesis.

    The natural language for using caches is just cache valued expressions made with  
operations $\vee$, $\wedge$, $/$, $//$ and $<$.  Here, for example, is a typical cache-valued expression
\begin{equation}
( X < ( ({\rm depth}:3,0) \vee Y/\{{\rm name}:*.py\} ) ) /\{\}
\end{equation}
where $X$ and $Y$ are caches and {\bf depth} and {\bf name} are data types.
In order to keep the user experience as simple as possible, we would like to simplify the language
syntactically and, if possible, hide the underlying lattice concepts in a more intuitive and familiar framework.
In view of this, introduce three reference caches:
``.'' (current working cache), ``$\sim$'' (home cache) 
and ``\verb @ '' (home caches of people known to you) and
introduce a set of caches with conventional short names like {\tt pwd}, {\tt cd} and {\tt ls}.  
Proceeding informally, replace

\begin{itemize}
\item{\tt X/\{\}} with {\tt X/};
\item{{\tt\{name:foo\}} with {\tt foo}};
\item{\tt(X < Y)/\{\}} with {\tt X Y} if {\tt X} is a named cache;
\item{{\tt (d,0)} with {\tt d}} and
\item{$\vee$ with {\tt ` '}}
\end{itemize}
to define a language called {\it egg shell} (the exact language definition can be found in Appendix A).
Using these replacements, and letting $X={\tt ls}$ and $Y={\tt .}$, equation 12
\begin{verbatim}
    ls depth:3 ./*.py
\end{verbatim}
has a syntactically reassuring resemblance to a UNIX shell command.  To see that the resemblance 
is not superficial, notice that the cache algebra also has built in concepts of ``execution'' and ``piping between shell
commands.''  Any cache $A$ may be ``executed'' by computing $A/$, in other words, by computing its contents.
For example, entering {\tt ls .} into the egg interpreter causes it to compute
\begin{verbatim}    (I < (~/shells/ls < .)/)/\end{verbatim} 
so that the contents of ($\sim${\tt /shells/ls < .)/} are put into 
the interpeter where they are explicitly written to terminal output.  When multiple shell commands are 
used, they effectively ``pipe their output into each other from right to left.''
For example, a typical egg interpreter command
\begin{verbatim}
    . < music < has ext:mp4 ex depth:3 web.google "free music"
\end{verbatim}
uses egg shell commands {\tt web.google} (which comes from the {\tt web} plug-in), {\tt ex} and {\tt has} (which come from the
{\tt yolk} plug-in) to cause the contents of a Google search for ``free music" ({\tt web.google}) to be expanded to depth 3 ({\tt ex depth:3}) 
and puts any mp4 files ({\tt has ext:mp4}) into a new cache  named ``music'' which is then stored in the current 
working cache ``.''.  For a more database-like example,
\begin{verbatim}
    count has tstart:class lines ~/egg/plugins/*//ext:py
\end{verbatim}
uses the {\tt lines} shell command to split any text data from the egg source code ($\sim${\tt /egg/plugins/*//ext:py})
into individual lines, selecting only those lines which begin with ``class'' ({\tt has tstart:class}) by 
extending {\tt text} data (which has trivial ordering) to {\tt tstart} data (which has prefix ordering).  As a result, 
this command counts the number of classes in the egg source code.

   Searching for the simplest distributed computing infrastructure, we find that the assumption that only one 
kind of object is involved combines with the need to do storage to imply that caches must 
at least form a distributive lattice.  Arguing that caches should be freely constructed
from the partially ordered set of all data, we come to an explicit free distributive 
lattice once the partially ordered set of all data is chosen.  Adding operators $/$ and $//$ 
for convenience and $<$ to allow caches to do more than storage, we have the ``essentially unique'' 
simplest candidate for a general purpose infrastructure.  Through experience with the
egg implementation, we find that caches as described here are able to conveniently represent
all the functional elements of what one wants in a distributed computing infrastructure including
files, directories, web sites, data bases, batch jobs, running processes, servers and many others.

From the perspective of a single user, the infrastructure resembles 
a database-integrated distributed operating system where UNIX-like shell commands pipe cache streams rather 
than byte streams and where the global set of all caches takes the place of the file system.
The normally difficult problem of cache coherence
in distributed file systems\cite{cachecoherence1,cachecoherence2,cachecoherence3} is relegated to
implementation details of the few cache subtypes which need it, such as banks.  Other than
caches with such special implementations, changes made to caches in one egg session are not reflected in other
egg sessions without an explicit {\tt save} command.  Each egg session is ``editing the internet'' so that
two users may modify the same cache, may see different results and only the last person to save is guaranteed
to see the same results in a new session.
Since caches are self-referential but not hierarchical, file system content smoothly combines with web content, 
bridging the gap between the Plan-9 like ``everything is a file'' view\cite{plan-9} and attempts to 
abstract server-based functionality such as web services\cite{web-services} and the semantic web\cite{semantic-web}.  
Compared with UNIX, more combinations of shell commands pipe into each other usefully and one seems to need fewer 
shell commands and fewer optional ``flags.''  Typically zero or one flags are needed per egg shell command.

\section{Servers, cache addressing and home caches}

    On a large scale, the infrastructure described here is a collection of independent processes representing
personal interactive sessions or servers.  Within a particular process, the other processes in the system
are represented by caches.  The most na\"{i}ve behavior of such caches would be to execute whatever cache is
put into them by the client. At the interprocess communication level, each server would
\begin{enumerate}
\item{Receive a serialized cache {\tt X}};
\item{Return the serialization of {\tt X/}}.
\end{enumerate}
Since cache execution is sufficient to do any operation, this would give the client complete control of the server.
Real servers use a simple variation of this where the returned results go through a {\tt server} egg shell command 
on the server side:
\begin{enumerate}
\item{Receive a serialized cache {\tt X}};
\item{Return the serialization of {\tt server X}};
\end{enumerate}
so that a server process appears to the client as a cache which yields {\tt server X} when {\tt X} is put into it. 
This formulation of server behavior leaves room for improvements to clients and servers over time while keeping 
the byte stream level protocol unchanged.
Although servers are defined here as executers, they are typically not directly used that way.
Instead, base class implementations of the
cache algebra use server cache execution to implement $\vee$, $\wedge$, $/$, $//$, and $<$ so that caches hosted by 
remove servers appear seamlessly.  Given a particular singleton cache $(d,X)$, its operations are implemented by a 
server if $d$ contains a server's personal public key, IP address and port number.  
Existing protocols like SSH and HTTP are included as artificial persons so as to fit in the same framework as the 
servers described here.

From the egg shell perspective, standard protocols and egg server caches appear uniformly and can be 
referred to using basic internet data such as
\begin{verbatim}
    cd http:www.bu.edu
\end{verbatim}
or
\begin{verbatim}
    ls {person:Alice,port:33366,host:physics.bu.edu,path:home}/
\end{verbatim}
where one refers to host names or IP addresses explicitly.  Alternatively, it is often more convenient 
to ignore low level addressing and to instead refer to caches relative to the ``home caches'' of persons, 
organizations or servers known to you.  A person, server or organization creating such a public/private key identity may also provide any 
egg shell expression defining their {\it home cache}. Home caches
may be simple web pages such as {\tt http:www.bu.edu} or may be more sophisticated such as
\begin{verbatim}
    first hasnt d:error test {http:www.bu.edu/~Alice}  {http:physics.bu.edu/~Alice}
\end{verbatim}
which provides the first of two web pages which are error-free.  A home cache may also be the output of an 
egg shell command so that the home cache can be modified whenever the corresponding plug-in software is updated.
The home caches of all persons and servers known to you appears in the ``$@$'' cache in egg shell, so one can
navigate by
\begin{verbatim}
    cd @/BostonUniversity
\end{verbatim}
or 
\begin{verbatim}
    cd @/Alice
\end{verbatim}
rather than by referring to explicit hosts or IP addresses.  In effect, home caches make a flexible common address space
not shared globally, but shared by all who have a common public key in their rolodex.

\section{Economics}

    A large scale infrastructure na\"{i}vely built along the lines described so far would have a number
of problems in common with the present-day internet.  For example, we have yet to propose a mechanism for 
sharing resources among multiple users in accordance with the wishes of the owners of the hardware.
In conventional terms, we also, so far, lack any mechanism for
authentication, access control, accounting or resource allocation and have made no mention of virtual 
organizations\cite{gsi}.  
Instead of treating these issues individually, we proceed in two steps.  First, introduce
a common hierarchical currency to explicitly represent the intentions of actors within the system.  
Next, treat problems like authentication and resource allocation as economic problems to be solved
by currency exchange.  Within this framework, improved solutions to these problems can appear over time
via improved cache designs without disrupting the infrastructure as a whole.  At the same time, 
currency exchange opens a new dimension for system improvement by designing caches to compete
with each other for currencies of interest.  
Over time, the design of such caches can be improved by using, for example, 
sophisticated auctions\cite{Parkes,Kang} and resource estimation techniques\cite{timeisright} 
so that the infrastructure  becomes increasingly internally efficient and increasingly responsive 
to the desires of people and organizations.

    Economics within the infrastructure is based on a hierarchical currency called the ``egg.''  Eggs appearing 
in the form of a data type called {\bf check} which can be minted by any user, transferred from user 
to user and which may be earned by caches. Each user has a bank cache which is used for storing 
checks and executing bank to bank transfers.  Possession of a valid check
\begin{equation}
     {\tt 100[Jan 1,2011]BostonUniversity.Alice.Bob}
\end{equation}
for example, cryptographically guarantees that the person associated with the
BostonUniversity private key minted (at least) 100 eggs, transferred these to Alice 
who then transferred 100 eggs to Bob.  Transfers can be gifts, as above, or payments denoted by an underline as in
\begin{equation}
     {\tt 100[Jan 1,2011]BostonUniversity.Alice.Bob.}\underline{{\tt server1}}
\end{equation}
where Bob has paid the entire 100 eggs to server1.  Depending on the intent of the owner of server1, the
server1 bank might be configured to accept currencies according to
\begin{center}
    \begin{tabular}{|l|l|l|}
         currency & who & preference  \\ \hline
	 BostonUniversity.*.Alice.? & anyone given BostonUniversity currency by Alice & 1000.0 \\
	 BostonUniversity.* & anyone possessing BostonUniversity currency & 50.0 \\
	 HarvardUniversity.* & anyone possessing HarvardUniversity currency & 40.0 \\
	 * & the general public & 1.0 \\ \hline
    \end  {tabular}
\end  {center}
which effectively controls access by listing acceptable currencies and implies that server1 will prefer to
earn BostonUniversity currency given by Alice over all others.  Note that server1 can safely grant access to 
anyone who Alice decides to give currency to without having to know who these people are in advance and 
without using a certificate authority.  Payments made to server1 are stored in the server1 bank cache which then 
transfers the check back to the payer until the check returns to its original creator and becomes a complete {\it receipt}.  

    To define the currency more explicitly, we presume a common cryptosystem with public keys, private keys
and digital signatures.  A {\it payload} is a string encoding a tuple consisting of a floating point 
currency denomination, a starting date, an expiration date, a tracking number and a payment bit.  A ``check''
is a signed tuple of strings ending with a public key.  The {\it owner of a check} is the owner of
this public key.  With these preliminaries, 
\begin{defin}
A check is either a payload/public key pair $(L,k)$ signed by anyone, or 
a payload/check/public key triplet $(L,c,k)$ signed by the owner of the check $c$.
\end{defin}
\noindent To illustrate the life cycle of a check, suppress the payload except for denomination and payment bit,
let $A$, $B$ and $C$ be be identities (public/private key pairs) and denote text $x$ signed by $A$ by $[x]_A$.
\begin{center}
\begin{tabular}{|l|l|l|}
         action & check & display \\ \hline
         A gives B 100 eggs & $[100,B]_A$ & $100A.B$ \\ 
         B pays C 50 eggs & $[50P,[100,B]_A,C]_B$ & $50A.B.\underline{C}$ \\ 
	 C returns check to B & $[50,[50P,[100,B]_A,C]_B,B]_C$ & $50A.B.\underline{C}.B$ \\
	 B returns check to A & $[50,[50,[50P,[100,B]_A,C]_B,B]_C,A]_B$ & $50A.B.\underline{C}.B.A$ \\ \hline
\end{tabular}
\end{center}
Examining the second step in more detail, $B$ begins with $[100,B]_A$, constructs $[50P,[100,B]_A,C]_B$ using 
the $B$ private key and sends the resulting string to $C$.  To preserve the initial 100 egg value, $B$ is required to both 
produce and save $[50,[100,B]_A,B]_B$, self-giving 50 eggs, and to destroy the original string $[100,B]_A$.  If the $B$ bank
follows this procedure, no new currency is generated or lost while making the payment.  Each new payload created by the $B$ bank keeps
the existing start and expiration dates and is given a unique tracking number so that the check can be identified after it returns
to the $A$ bank.  In the last table entry, the check returned to $A$ has completed its circuit and is thus a {\it receipt}.

Although the presumed integrity of digital signatures prevents checks from being forged, it is clear that some behavior that
we need is only enforced by the software of the bank caches.  In particular, handling of checks must satisfy banking rules:
\begin{enumerate}
\item{Other than minting a new check, all check operations must preserve total value;}
\item{Checks with a payment bit set may not be used for additional payments;}
\item{Cashed checks recieved by a bank must be transferred to the previous bank in the check sequence.}
\end{enumerate}
\noindent Determined users could, however, violate the banking rules by modifying their own bank cache
software or by copying checks outside of the system. Although such rule violations cannot be prevented, 
they can be detected after the fact by examining receipts and tracking numbers.  It is not clear to us 
whether such problems should be discouraged, prevented by putting all banking caches on trusted servers or 
detected and punished.

     Unlike previous attempts to combine economics and distributed
computing\cite{other1,other2,other3,other4,other5}, egg currency is compatible with a wide variety of
independent economic behaviors ranging from ignoring currency entirely to equally sharing among friends to traditional
fixed monthly allocations to commercial situations where earning currency results in real financial gain.
Unlike grid systems where virtual organizations are layered on top of distributed computing\cite{gsi},
virtual organizations here appear organically: users, servers and virtual organizations are the same thing.  
Take, for example, an organization like the ATLAS high energy physics experiment.  ATLAS needs a computing infrastructure
which includes the 2500 physicists and 169 universities and laboratories in the collaboration\cite{ATLAS}.
Within the infrastructure proposed here, ATLAS would be a single user like everyone else with a single private key, rolodex and bank.  
If, for example, ATLAS decides to spend 90\% of it's worldwide computing resources on Higgs searches and the remaining 10\% on 
Supersymmetry for the next month, ATLAS would mint two checks {\tt 90[Mar-1-2009]ATLAS.Higgs} and {\tt 10[Mar-1-2009]ATLAS.Supersymmetry} sent to the Higgs user and the 
Supersymmetry user respectively and continuing recursively until one reaches the people who directly spend the currency.  
Worldwide resources respond to the newly minted checks because ATLAS collaborators and computing centers want to earn 
the ATLAS currency because they want to contribute to the project.  The ATLAS person can have this effect without
having to know exactly who is in the Higgs group or who is using what resources when.  At the largest scale, the currency makes a
global certificate authority unnecessary (note that a user who gave any requested amount to any acceptable 
user is the same thing as a certificate authority).  The egg currency, then, is intended to be used at all scales, from 
global applications to organizational applications to server authentication to setting access controls on a single file.   
Although our current egg implementation includes egg currency as defined above, only the most basic economic interactions
have been used so far: user directed gifts and payments, two-way authentication by currency exchange and
price-based access control.  In spite of the fact that much theoretical progress has been made as to how caches can compete with
each other in auctions\cite{Kang}, basic questions remain as to how economically activated caches should 
be designed, how they should be collectively used and what properties such collections might have.  Basic
questions also remain in the macroeconomics of the egg currency: how should one optimally deal with bad behavior 
such as banking rule violations, inflation and caches which accept payment but do not provide expected service?
It is important to design banks so that easy to understand monetary policies can be set and then executed
over time with minimal direct user intervention. Since the economic system has an important social component, 
it is hard to make do without large scale testing involving a reasonable number of users.  
This has not yet been done.

\section{Implementation}

    Starting with an attempt to make an infrastructure from only one kind of object, we 
have arrived at a concrete system consisting of caches, a cache algebra, a language and a protocol
which is essentially uniquely determined in the sense described in section 2.  Adding the egg 
currency provides a framework for authentication, access control, resource allocation and 
a proposal that system improvements over time come from designing caches to compete with each other
to earn eggs.  
In order to understand whether this framework really provides everything that one expects of a general purpose distributed
computing infrastructure, we have made an implementation of the infrastructure called {\it egg} written in Python\cite{python}.  
Caches, data types, $D$ and singleton caches are implemented as Python classes which are organized
into ``plug-ins.''  Egg includes plug-ins for the core system: {\bf yolk} and {\bf egg}, 
for common protocols: {\bf ssh}\cite{ssh}, {\bf web}\cite{web}, {\bf gsi}\cite{gsi}, {\bf local}\cite{local}, 
for batch schedulers: {\bf pbs}\cite{pbs}, {\bf lsf}\cite{lsf}, for commmon unix software {\bf linux}\cite{linux}, 
{\bf rpm}\cite{rpm}, {\bf apache}\cite{apache}, 
for mysql based servers {\bf mysql}\cite{mysql}, {\bf lfc}\cite{lfc}, and for the ATLAS high energy physics 
experiment\cite{ATLAS}, {\bf pacman}\cite{pacman}, {\bf panda}\cite{panda}, {\bf dq2}\cite{dq2}, {\bf atlas}\cite{ATLAS}.
When imported, each plug-in contributes new data types, extensions, singleton cache subtypes 
and egg shell commands to the system.  Most of the plug-ins are small (typically a few hundred lines of Python) and
are easy to write without detailed knowledge of the core system.  Figure 1, for example, shows the complete 
Python code required to define an egg shell command {\bf lines} which splits text into caches containing one
line of text per cache.  Including all 24 plug-ins, there are 194 egg shell commands, 439 cache subtypes, 
845 data types and 330 extensions in the system. 
\begin{figure}[htbp]
\begin{center}
  \leavevmode
  \includegraphics[width=0.85\textwidth,height=0.25\textheight]{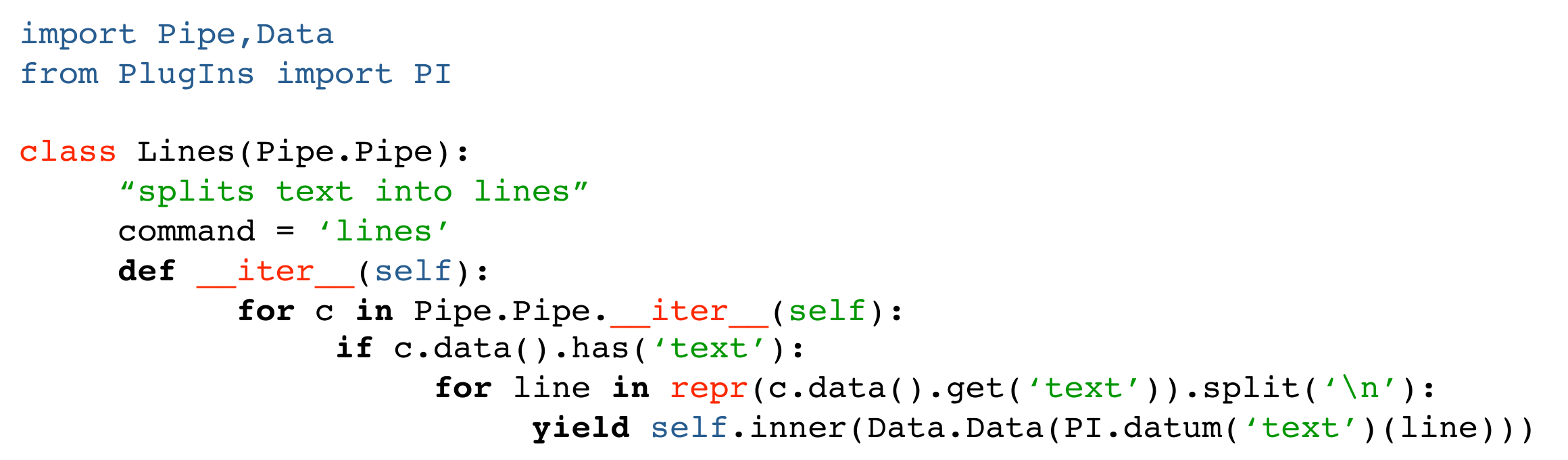}
  \caption{{\it Python source code for a functioning egg shell command called {\bf lines}. {\bf Pipe} is
  a singleton cache subclass where {\bf put}$(d,X,Y)$ is a function of $Y$ only.  {\bf Data} is the Python version
  of the universal partially ordered set of data $D$. {\bf PI} is the plug-in manager allowing the creation of
  new objects of selected basic data types.  The {\bf command=`lines'} data member is picked by the plug-in 
  manager causing a shell command named {\bf lines} to automatically appear in the interpreter.   
  For each cache {\bf c} that extends to {\bf text} data, the inner loop lazily produces one cache for
  each line of text using the {\bf self.inner} function which creates caches from data.  Most singleton cache subtypes such as {\bf Lines} need only provide a subclassed Python iterator.}}
\end{center}
\end{figure}
    Extensions as defined in section 2 play important roles in the system which we did not
entirely anticipate.  These not only effectively save space and time via lazy evaluation, but they also simplify 
the introduction of new cache types and new egg shell commands.  In using the infrastructure, we find that concious 
considerations of extensions or lazy evaluation isn't necessary and extensions work smoothly below the level of normal user 
awareness. The {\tt stat} command is an example of unexpected benefits: {\tt stat} is able to collect statistics for 
any data type which is an abelian group as well as a meet semilattice.  It can, for example, compute 
statistics for file sizes by summing the {\bf size} data which are put into it.  
Plain text data, on the other hand is not an abelian group and is ignored by {\tt stat}. 
The following, for instance,
\begin{verbatim}
     stat d:size .//
\end{verbatim}
sums the file sizes of all files in {\tt .//}, but, in addition, it sums all data types which extend from size, producing
useful additional statistics about size (such as a histogram of sizes) without having to know that these statistics 
exist ahead of time.  A typical use of this feature would be something like
\begin{verbatim}
     stat d:pbs pbs.jobs ssh:atlas.bu.edu
\end{verbatim}
which produces a summary of all PBS-related\cite{pbs} data produced from PBS jobs running on {\tt atlas.bu.edu}.
We find that a substantial part of the system can be written just by writing extensions. 

    One of the most useful features of the system comes from the simple observation that a text file containing
egg shell code can either be executed as a script or can be interpreted as defining the contents of the text file itself
 so that one can, for example, ``cd into a script.''  Files which are interpreted as contents rather than as a script 
 to be executed are called ``hatches.''  A typical hatch, for example, would be a file called {\tt hosts.hatch} containing:
\begin{verbatim}
     #
     #  NET2 worker nodes
     # 
     show d:linux.host,linux.load,linux.uptime
     linux.hosts @/NET2/Boston/nodes/
\end{verbatim}
When a cache is created for the {\tt hosts.hatch} file, it contains the output of the indicated {\tt linux.hosts} egg shell
command which produces one cache representing each of the Boston University worker nodes in the U.S. ATLAS Northeast Tier2 Center\cite{net2}.
The line {\tt show d:name,linux.host,linux.load,linux.uptime} defines default display options.  The net effect of this is that if you 
cd into {\tt hosts}, you see the linux hosts each with load and uptime information without having to remember that the
{\bf linux} plug-in contains the {\tt linux.hosts} command or having to remember how to use it.  Hatches make it easy to
blend caches defined by standard protocols like HTTP and SSH with egg server caches with the output of egg shell commands and
with other caches.  This strongly re-enforces the impression of a uniform infrastructure spanning all content and functionality.

    The success of a uniform cache oriented view of all infrastructure content means that one quickly loses track of which
protocols are used by which caches and which caches are hatches and which are not.  This, of course, is the desired result, but 
it also means that any protocol-level errors which occur will be incomprehensible to most users.  In general, errors are treated 
in egg by producing error data rather than by raising exceptions in the interpreter.  By default, high-level cache oriented
error data are visible in the interpreter.  This allows users to notice data of the {\bf error} type, collect them and analyze 
them using the same tools which are used for analyzing any other kind of data.  Creation of {\bf error} data is typically 
accompanied by created {\bf errorbase} (low level error) and {\bf traceback} (Python traceback) data which can be also 
treated and analyzed like any other kind of data.  We also found it convenient to introduce a logging cache, which is 
particularly useful for collecting statistics about running egg sessions.
\begin{figure}[htbp]
\begin{center}
  \leavevmode
  \includegraphics[width=0.85\textwidth,height=0.4\textheight]{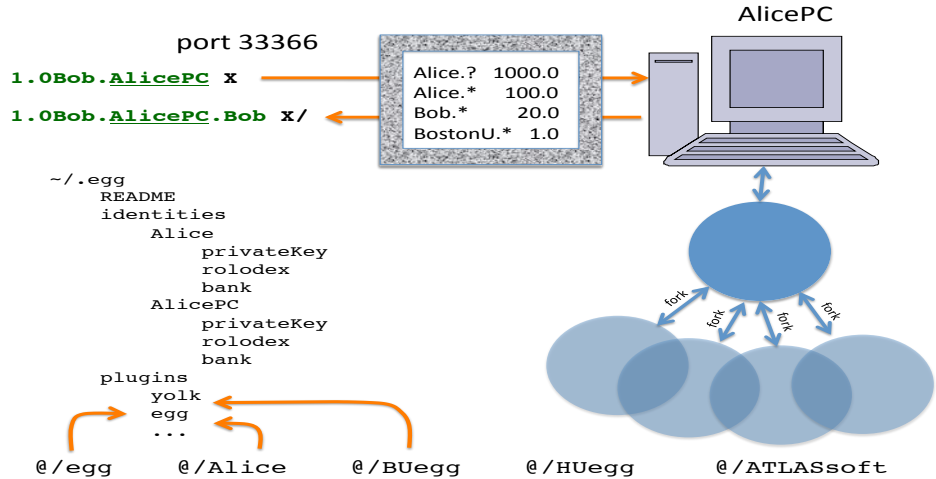}
  \caption{{\it A schematic of a simple egg server as installed on Alice's PC.}}
\end{center}
\end{figure}

    In addition to interprocess communication with servers, egg shell commands in this implementation are able to 
execute in parallel by forking and returning results from multiple child processes, allowing effective use of 
multi-core machines.  An ``execution network'' keeps track of hosts which are reachable by protocols allowing Unix shell execution.  
By composing Unix shell commands, egg can effectively reach further than by SSH from a single account and can 
transparently use the entire network for extensions and can transparently use software like MySQL\cite{mysql} 
and Globus\cite{gsi} even if the software is not available where the egg interpreter is actually running.    

    The egg infrastructure includes simple servers and a basic economic system which is capable of
basic currency and banking operations, authentication and access control. Figure 2 shows a schematic 
of a simple egg server installed on Alice's PC.  AlicePC accepts incoming messages either from someone
who Alice has directly given currency to (providing {\bf Alice.?}), or from someone holding any Alice, Bob
or Boston University currency (providing {\bf Alice.*}, {\bf Bob.*} or {\bf BostonU.*}).  The server
values these currencies with respect to each other in the ratios 1000.0/100.0/20.0/1.0 respectively as 
indicated in the figure with the idea that more advanced server designs would act to prefer earning higher valued
currencies.  As indicated, Bob communicates with AlicePC by sending a serialized
cache $X$ containing a payment {\tt 1.0Bob.}\underline{\tt AlicePC} and a serialized cache to be executed on port 33366.  
AlicePC returns a receipt {\tt 1.0Bob.}\underline{\tt AlicePC}{\tt .Bob} to Bob (thus completing a two way authentication) 
and the serialization of $X/$, the requested execution.  The server operates by forking a process to handle 
each new connection.  Alice's egg software installation is indicated on the left side of the figure.  
The figure indicates plug-ins being dynamically imported from the home caches of trusted 
persons {\bf egg}, {\bf Alice}, {\bf BUegg}, {\bf HUegg} and {\bf ATLASsoft}.

    A number of applications of the system related to the U.S. ATLAS Northeast Tier 2 computing 
facilities at Boston University and Harvard University\cite{net2}.  A typical simple ATLAS-specific operation is 
\begin{verbatim}
     atlas.unsatisfied @/NET2/Harvard/nodes
\end{verbatim}
which uses the {\tt atlas.unsatisfied} command from the {\bf atlas} plug-in to test whether all of the 
Harvard worker nodes have the RPMs and libraries necessary for ATLAS Monte Carlo simulation jobs to run successfully.
A typical more advanced command
\begin{verbatim}
     stat d:unix tr din:text dout:path hasnt tstart:srm split sep:atlas.bu.edu \
           tr din:lfc.pfname dout:text @/NET2/DDM/lfc//since:yesterday
\end{verbatim}
uses five egg shell commands {\tt stat}, {\tt tr}, {\tt hasnt}, {\tt split}, {\tt tr} to test if the path 
entries which have arrived since yesterday in an ATLAS-specific MySQL database and server called ``LFC'' really
exist in the local file system.  Reading from right to left, entries in the LFC database includes dates which 
extend to the data type {\bf since}, which is a time interval partially ordered by inclusion.  As a result,
{\tt @/NET2/DDM/lfc//since:yesterday} contains all database entries which have arrived since {\tt yesterday} 
(defines as the time interval from 24 hours ago to now in univeral time coordinates).  The {\tt tr} shell 
command extracts any {\tt lfc.pfname} data containing the URLs of files which are supposed to 
exist in our storage element.  The shell commands {\tt split} and {\tt hasnt} are then used to manipulate 
the text value of the {\tt lfc.pfname} so as to extract the UNIX path as text data being emitted by 
the {\tt hasnt} command.  The second {\tt tr} command then translates the text path into {\bf path} data 
so that it is interpreted aa the path of a file in the local file system.  Finally, {\tt stat} summarizes
all of the UNIX level information about these files such as whether they exist, their size, creation/modification/access
times, ownership and access control bits.
    
    The egg system is developed enough to give a clear impression of the proposed infrastructure in 
a wide variety of applications.  We find the situation encouraging in that all the needed elements of the infrastructure
easily fit in the framework where everything is a cache and all added functionality is provided by 
cache put operations.  Although the economic system has only been tested in the limited way explained in section 5, 
we have not found any needed feature which does not easily fit within the framework described here.

\section{Appendix: The egg shell language}
Egg shell is a very simple language with no variables, subprograms, flow control, data structures, classes, assigments, functions or exceptions.  
It has only four binary operations: concatenation, /, // and $<$.  Power and flexibility comes from an available set of
egg shell commands.  It is convenient to directly define the language as a recursive 
partial function $\pi$ mapping strings to caches.  Strings which are not in the domain of $\pi$ are, 
by definition, not legal egg shell programs.
The compiler $\pi$ is defined in terms of another partial function $\delta$ mapping strings to data (elements of $D$).
Both $\pi$ and $\delta$ are defined by sequences of partial functions where the action of the sequence on a 
string $s$ is defined to be action of the first item in the sequence for which $s$ is in the domain of the partial function.
First, define $\delta$ by the following:

\begin{center}
    \begin{tabular}{|l|l|l|}
    	 name & pattern & result \\ \hline
	 left blank & $'\ 'x$ & $\delta(x)$ \\ \hline
	 right blank & $x'\ '$ & $\delta(x)$ \\ \hline
	 comma & $x','y$ & $\delta(x)\wedge\delta(y)$ \\ \hline
         curly & $\{x\}$ & $\delta(x)$ \\ \hline
	 curly2 & $x\{y\}$ & $\{name:x\}\wedge\delta(y)$ \\ \hline
	 datum & $x:y$ & ${x:y}$ \\ \hline
	 maximum & $''$ & $\{\}$ \\ \hline
    \end  {tabular}
\end  {center}

\noindent Given $\delta$, we can define the compiler $\pi$.  Compilation is always done within
the context of a particular set of available egg shell commands.  Let {\bf shell} map egg shell command names to their 
corresponding caches.  Let {\bf DOT}, {\bf TILDE} and {\bf AT} be the caches corresponding to the three distinguished named 
caches ``., $\sim$ and \verb @ '' respectively and define $\pi$ by 

\begin{center}
    \begin{tabular}{|l|l|l|}
    	 name & pattern & result \\ \hline
	 lines & $x\ {\bf newline}\ y$, leftmost & $\pi(x)\vee\pi(y)$ \\ \hline
	 left blank & $'\ 'x$ & $\pi(x)$ \\ \hline
	 right blank & $x'\ '$ & $\pi(x)$ \\ \hline
	 shell command & $C\ x$ for command C & $({\bf shell}(C)<\pi(x))/\{\}$ \\ \hline
	 put & $x\ <\ y$, rightmost & $\pi(x)<\pi(y)$ \\ \hline
	 lub & $x\ y$, leftmost & $\pi(x)\vee\pi(y)$ \\ \hline
	 parenthesis & $(x)$ & $\pi(x)$ \\ \hline
	 slashes & $x/d$ or $x//d$, rightmost & $\pi(x)/\delta(d)$ or $\pi(x)//\delta(d)$ respectively \\ \hline
	 current working cache & $'.'$ & {\bf DOT} \\ \hline
	 context & $'\sim'$ & {\bf TILDE} \\ \hline
	 friends & $'{\bf @}'$ & {\bf AT} \\ \hline
	 singleton & $x$ & $(\delta(x),0)$ \\ \hline
    \end  {tabular}
\end  {center}
 
\noindent In text files, backslash line continuation and ``\#'' characters indicating comments are handled in the usual way.


\begin{thebibliography}{10}
\bibitem{Egg} J. Brunelle, P. Hurst, J. Huth, L. Kang, C. Ng, D.C. Parkes, M. Seltzer, J. Shank and S. Youssef,
{\it Egg: An Extensible and Economics-Inspired Open Grid Computing Platform}, in:
Proc. 3rd Int. Workshop on Grid Economics and Business Models (GECON'06), Singapore, 2006.
\bibitem{lattice} A {\it lattice} is a set with two associative, commutative, idempotent binary operations $\vee$ (``join'') and 
$\wedge$ (``meet'') satisfying $x\vee(x\wedge y)=x\wedge(x\vee y)=x$ for all $x$ and $y$ in the set.  A {\it distributive lattice}
also satisfies $x\wedge(y\vee z)=(x\wedge y)\vee (x\wedge z)$ for all $x,y,z$.  See {\it Notes on Lattice Theory} by J.B. Nation at {\tt http://www.math.hawaii.edu/}$\sim${\tt jb/books.html}.
\bibitem{category} R. Geroch, {\it Mathematical Physics}, University of Chicago Press, 1985;
M. Barr and C. Wells, {\it Category Theory for Computing Science}, Prentice Hall, 1995; 
S. MacLane, {\it Categories for the Working Mathematician}, Springer, 1998.
\bibitem{antichain} A subset $A$ of a partially ordered set is a {\it chain} if it is maximally ordered: if $x\leq y$ or $y\leq x$ holds for
each $x,y\in A$.  Subset $A$ is an {\it antichain} if it is minimally ordered: if $x\leq y\Leftrightarrow x=y$.  A function $f:X\rightarrow Y$
is {\it order preserving} or {\it monotonic} if $x\leq x'$ implies $f(x)\leq f(x')$.
\bibitem{uparrow} In other words, $A^{\uparrow}$$\stackrel{d}{=}$$\{a\in A$ such that $a\leq a'\Rightarrow a=a'$ for all $a'\in A\}$.
\bibitem{meetsemidirectsum} If $A$ and $B$ are meet semilattices with minimum elements, the meet semilattice
direct sum $A\oplus B$ is the set of subsets of the disjoint union of $A$ and $B$ reduced by applying 
meet operations wherever defined.  Thus, an element of $A\oplus B$ contain at most one element of $A$ and at most one element 
of $B$.  
\bibitem{cachecoherence1} M.N. Nelson, B. Welch and J.K. Ousterhout, 
{\it Caching in the Sprite Network File System}, 
ACM Transactions on Computer Systems, volume 6, 1988.
\bibitem{cachecoherence2} T. Anderson, M. Dahlin, J. Neefe,
D. Roselli and R. Wang, Serverless network file systems, in
{\it Proceedings of the 15'th Symposium on Operating Systems
Principles}, ACM, December 1995.
\bibitem{cachecoherence3} C. Gray and D. Cheriton, 
{\it Leases: an efficient fault--tolerant mechanism for 
distributed file cache consistency}, in: SOSP '89: 
Proceedings of the twelfth ACM symposium on operating system
principles, ACM, New York, 1989.
\bibitem{plan-9} R. Pike, D. Presotto, S. Dorward, B. Flandrena, K. Thompson, H. Trickey
and P. Winterbottom, {\it Plan 9 from Bell Labs}, Computing Systems, Volume 8, number 3,
cite-seer.istpsu.edu/pike95plan.html, 1995.
\bibitem{web-services} D. Booth, H. Haas, F. McCabe, E. Newcomer, M .Champion,
C. Ferris and D. Orchard, {\it Web Services Architecture}, W3C Working Group Note 11, {\tt http://www.w3.org/TR/ws-arch/}, February, 2004.
\bibitem{semantic-web} {\it W3C Semantic Web Activity}, {\tt http://www.w3.org/2001/sw/}, 2001.
\bibitem{Parkes} R.K. Dash, N.R.Jennings and D.C.Parkes, {\it Computational-Mechanism Design: A Call to Arms},
IEEE Intelligent Systems, 18:40-47, 2003.
\bibitem{Kang} L. Kang and D.C. Parkes, {\it A Decentralized Auction Framework to Promote Efficient Resource
Allocation in Open Computational Grids}, in: Proc. ACMEC Joint Workshop on The Economics of 
Networked Systems and Incentive-Based Computing, San Diego, 2007.
\bibitem{timeisright} B.N. Chun, P. Buonadonna, A. AuYoung, C. Ng, D.C. Parkes, J. Shneidman, A.C. Snoeren and
A. Vahdat, {\it Mirage: A Microeconomic Resource Allocation System for SensorNet Testbeds}, in:
Proc. of 2nd IEEE Workshop on Embedded Networked Sensors (EmNetsII), 2005.
\bibitem{other1} C.A. Waldspurger, T. Hogg, B. Huberman, J.O. Kephart and 
W.S. Stornetta, {\it Spawn: A Distributed Computational Economy},
IEEE Trans. on Software Engineering 18:103-117, 1992.
\bibitem{other2} R. Woldki, J. Brevik, J. Plank and T. Bryan, {\it Grid Resource 
Allocations and Control Using Computational Economies}, in: 
Grid Computing: Making the Global Infrastructure a
Reality, F. Berman, G. Fox and T. Hey (Eds.) pp.747-772, Wiley and Sons, 2003.
\bibitem{other3} R. Buyya, D. Abramson and J.  Giddy,
{\it NimrodG: An Architecture of a Resource Management and 
Scheduling System in a Global Computational Grid}, in:
Proc. of the 4th Int. Conf. on High Performance
computing in Asia-Pacific Region, pp.283-289, Beijing, China, 2000.
\bibitem{other4} K.Lai, B. Huberman and L. Fine, {\it Tycoon: A Distributed
Market-based Resource Allocation System}, Technical Report, Hewlett Packard, No. cs.DC/0404013, 2005.
\bibitem{other5} I.E. Sutherland, {\it A futures market in computer time},
Commun. ACM, 11:449-451, 1968.
\bibitem{ATLAS} {\it The ATLAS Experiment}, {\tt http://atlas.cern.ch}.
\bibitem{gsi} I. Foster and C. Kesselman, {\it The Grid 2: Blueprint for 
a New Computing Infrastructure}, Morgan Kaufmann, November, 2003.
\bibitem{python} The Python language, {\tt http://www.python.org}.
\bibitem{ssh} D.J. Barrett, R.E. Silverman and R.G. Byrnes, {\it SSH:The Secure Shell}, O'Reilly 2005.
\bibitem{web} The world wide web, {\tt http://www.wc3.org/}.
\bibitem{local} {\bf local} is the plug-in handling the protocol of local files and directories.
\bibitem{pbs} Portable Batch System (PBS), {\tt http://www.pbsgridwords.com/}.
\bibitem{lsf} Load Sharing Facility, Platform Computing Corporation, {\tt http://www.platform.com/Products/platform-lsf/}.
\bibitem{linux} The Linux operating system, {\tt http://www.linux.org/}.
\bibitem{rpm} RedHat package manager, {\tt http://rpm.org/}.
\bibitem{apache} The Apache Software Foundation, {\tt http://www.apache.org/}.
\bibitem{mysql} {\it The MySQL database}, {\tt http://www.mysql.com}.
\bibitem{lfc} Local file catalogue, {\tt http://www.gridpp.ac.uk/wiki/LCG\_File\_Catalog}.
\bibitem{pacman} Pacman package manager, {\tt http://physics.bu.edu/pacman}.
\bibitem{panda} T. Maeno, {\it PanDA: distributed production and distributed analysis system for ATLAS}, 
J.Phys, Conf.Ser.{\bf 119} 062036, 2008.
\bibitem{dq2} M. Branco, D. Cameron, B. Gaidioz, V. Garonne, B. Koblitz, M. Lassnig, R. Rocha, P. Salgado and T. Wenaus, 
{\it Managing ATLAS data on a petabyte-scale with DQ2}, J.Phys., Conf.Ser.{\bf 119} 062017, 2008.
\bibitem{net2} The U.S. ATLAS Northeast Tier 2 center, {\tt http://atlas000.bu.edu/NETier2/wiki/index.php/}.

\end{thebibliography}
\end{document}